\documentclass[a4paper,12pt,reqno]{amsart}
          \usepackage{amssymb}
	  \usepackage{amsmath}
          \usepackage{amsfonts}
          \usepackage[english]{babel}
          \usepackage[utf8]{inputenc}

\usepackage[margin=2.5cm]{geometry}

 \usepackage[unicode,colorlinks,plainpages=false,hyperindex=true,bookmarksnumbered=true,bookmarksopen=false,pdfpagelabels]{hyperref}
 \hypersetup{urlcolor=cyan,linkcolor=blue,citecolor=red,colorlinks=true}
 \usepackage{color}

\newtheorem{thm}{Theorem}[section]

\newtheorem{lem}[thm]{Lemma}
\newtheorem{prop}[thm]{Proposition}

\newtheorem{rem}[thm]{Remark}

\numberwithin{equation}{section}

\newcommand{\be}{\begin{equation}}
\newcommand{\ee}{\end{equation}}
\newcommand{\bea}{\begin{eqnarray}}
\newcommand{\eea}{\end{eqnarray}}
\newcommand{\ba}{\begin{aligned}}
\newcommand{\ea}{\end{aligned}}

\def\1{{\boldsymbol 1}}                     %
\def\cH{{\mathcal H}}                       %
\def\tr{\mathrm{tr}}                        %
\def\diag{\mathrm{diag}}                    %
\def\ri{{\rm i}}                            %
\def\bC{\mathbb{C}}                         %
\def\bR{\mathbb{R}}                         %
\def\bP{\mathbb{P}}                         %
\def\cF{{\mathcal F}}                       %
\def\reg{\mathrm{reg}}                      %
\def\red{\mathrm{red}}                      %
\def\dt {\left.\frac{d}{dt}\right|_{t=0}}   %
\def\rank{{\mathrm{rank}}}                  %
\def\ad{\mathrm{ad}}                        %
\def\fN{\mathfrak{N}}                       %
\def\bZ{\mathbb{bZ}}                        %
\def\bT{\mathbb{T}}                         %
\newcommand\br[1]{\{ #1 \}}                 %
\def\cC{\mathcal{C}}                        %
\def\cS{\mathcal{S}}                        %
\def\cZ{{\mathcal Z}}                       %
\def\cK{{\mathfrak k}}                      %
\def\cP{{\mathfrak p}}                      %
\def\fg{{\mathfrak g}^\bC}                  %
\def\cT{{\mathfrak t}}                      %
\def\cA{{\mathfrak a}}                      %
\def\fG0{{{\mathfrak g}_0}}                 %

\begin{document}

\title{On the maximal superintegrability of \\ strongly isochronous Hamiltonians }

\maketitle

\begin{center}

L. Feh\'er${}^{a,b}$

\medskip
${}^a$Department of Theoretical Physics, University of Szeged\\
Tisza Lajos krt 84-86, H-6720 Szeged, Hungary\\
e-mail: lfeher@physx.u-szeged.hu

\medskip
${}^b$HUN-REN Wigner Research Centre for Physics\\
 H-1525 Budapest, P.O.B.~49, Hungary

\end{center}

\begin{abstract}
We study strongly isochronous Hamiltonians  that generate periodic time evolution with the same basic period
for a dense set of initial values.
We explain  that all such Hamiltonians are maximally superintegrable, and
 show that if the system is
subjected
to Hamiltonian reduction based on  a compact symmetry group
and certain conditions are met, then the
reduced Hamiltonian is strongly isochronous with the original basic period.
We utilize these  simple observations for demonstrating the maximal superintegrability
of rational spin Calogero--Moser type models in  confining harmonic  potential.
\end{abstract}

 \tableofcontents

\newpage

\section{Introduction}
\label{S:1}

 Let us recall that a  Hamiltonian is called superintegrable (or degenerate integrable) if
the number of its independent constants of motion is larger than half the dimension $D$ of the phase space,  here assumed
to be a symplectic manifold. Maximal superintegrability is the extreme case when the functional
dimension of the ring of constants of motion equals $(D-1)$, and this definition
also applies to Poisson manifolds.
Superintegrable Hamiltonian systems have a long history and represent an
active research subject.
For background and recent activities, we refer to \cite{CJRX,Ev,FG,WGMPnew,MPV,Nekh,Re1,Re2,WT} and references therein.

Our motivation for the present work arose from the paper \cite{BM}, in which the authors enquired about
the degenerate integrability of the rational spin Calogero--Moser model in a confining harmonic potential.
They were concerned with the quantum mechanics of the model, but here we focus
on the  corresponding classical Hamiltonian
\be
\cH_{\mathrm{spin}}(q,p, \zeta) =\frac{1}{2}\sum_{i=1}^n (p_i^2 + \omega^2 q_i^2) +
\frac{1}{2} \sum_{ i\neq j} \frac{\vert \zeta_i \zeta_j^\dagger\vert^2}{(q_i - q_j)^2},
\label{1}\ee
where $q_i, p_i$ $(i=1,\dots, n)$ are Darboux variables and the $\zeta_i$ are multicomponent complex row vectors
of fixed length ($\zeta_i \zeta_i^\dagger = c >0$),  whose
phases can be arbitrarily changed by gauge transformations.
If the $\zeta_i$ are just numbers, then
$\vert \zeta_i \zeta_j^\dagger\vert$ becomes a coupling constant, and \eqref{1} gives
the celebrated rational Calogero--Moser model (for $\omega=0$) and its confining variant (for $\omega \neq 0$), whose maximal
superintegrability is well known \cite{Ad,Wo1}.

For most physical applications the $q_i$ and the $p_i$ are real variables, associated with particles moving on the line,
but the complex holomorphic version of the model also received a lot of attention \cite{GH,Wil}.
 In the complex  case\footnote{In this case $\vert \zeta_i \zeta_j^\dagger \vert^2$ is replaced by $(\zeta_i \eta_j)(\zeta_j \eta_i)$ where
 the $\zeta$'s and the $\eta$'s are independent complex row and column vectors.}
 maximal superintegrability was recently proven in  \cite{FG}.
Their proof starts with the $\omega=0$ case, and then uses  complex analyticity to establish the functional independence
of some explicitly given constants of motion for generic values of the parameter $\omega$.
This proof does not apply  directly to the real case that interests us.

Our basic observation is that in the real case the maximal superintegrability, for arbitrary $\omega \neq 0$,  is an immediate
consequence of the `strongly isochronous' nature of the Hamiltonian \eqref{1}.
In this paper, a Hamiltonian is called strongly isochronous if
 all solutions of the Hamiltonian evolution equation
are periodic, with the same basic period $T>0$ for an open dense set of initial values.
Up to an obvious re-scaling, this means that the Hamiltonian is the moment map for a $U(1)$ action on the phase space.
 The simplest example  of such systems is the $N$-dimensional isotropic harmonic oscillator.

It is known to experts that all `strongly isochronous' Hamiltonians are
  maximally superintegrable.
  Without proof,  the claim  appears, for example, in Calogero's book \cite{Cal}.
In section \ref{S:2},  we deduce this statement from a known
property of invariants of compact Lie groups acting on smooth manifolds.
  Then, in section \ref{S:3}, we explain  that under certain
  conditions the Hamiltonian reductions of strongly isochronous systems remain strongly isochronous, with the original basic period.
The three subsequent sections are devoted to examples of strongly isochronous systems
obtained as reductions of higher dimensional isotropic harmonic oscillators.
In section \ref{S:4}, we deal with generalizations of the model \eqref{1} in which the spin variables belong to the dual
space of a compact simple Lie algebra, reduced by the maximal torus of the corresponding Lie group.
In section \ref{S:5}, we return to the Hamiltonian \eqref{1}, and expound its maximal
superintegrability.
The pair potentials featuring in \eqref{1} show that this model is associated with
a root system of type $A_{n-1}$, and in section \ref{S:6} we teat an analogous model
corresponding to the root system $B_n$, as given by the Hamiltonian in \eqref{HB2}.
Finally, in section \ref{S:7}, we  briefly discuss  open problems and draw attention to other examples
of confining  spin Calogero--Moser type models that emerge from Hamiltonian reduction.

We finish this introduction with the terminological remark that the confining spin Calogero--Moser models
are called harmonic spin Calogero--Moser models in \cite{FG}, and it would be also fitting
to refer to them as  spin Calogero--Moser oscillators.

 \section{Preliminaries and a general result}
 \label{S:2}

To begin, consider a connected $C^\infty$  Poisson manifold $(M,\br{-,-})$ with a distinguished real function $H \in C^\infty(M)$.
By definition, the  Hamiltonian system  $(M,\br{-,-}, H)$ (or just the Hamiltonian $H$) is
\emph{maximally superintegrable} if the commutant of $H$ in the Poisson algebra $C^\infty(M)$, given by
\be
\{ F \in C^\infty(M) \mid \{F,H\}= 0 \},
\ee
has functional dimension $\dim(M) -1$.
That is,  the points $x\in M$ for which
the exterior derivatives of the time independent, smooth first integrals of the Hamiltonian vector field $V_{H}$ span
a co-dimension $1$ subspace of $T_x^*M$ form a dense open submanifold of $M$.
Note in passing that in our convention the Poisson bracket and the Hamiltonian vector field are related by
$\{F, H\} = dF(V_H)$, $\forall F,H\in C^\infty(M)$.

Next, take a Hamiltonian $H\in C^\infty(M)$ and suppose that
 the Hamiltonian vector field $V_H$ is complete, i.e.,
the solution of the initial value problem
\be
\frac{d \phi_t(x)}{dt} = V_H(\phi_t(x)), \quad \phi_0(x)=x
\label{21}\ee
is defined for all $t\in \bR$ and $x\in M$.
Then, the smooth mapping
\be
\phi: \bR \times M \to M, \qquad \phi(t,x): = \phi_t(x),
\label{22}\ee
yields  an action of the additive group $\bR$ on $M$.
We call the Hamiltonian $H$ (or equivalently the  system $(M,\br{-,-},H)$)
 \emph{strongly isochronous } if the curve $\phi_t(x)$ is periodic in $t$ with the same period $T>0$ for all
  $x$ from a dense open submanifold $M_*$ of $M$.
In other words, we assume the existence of a dense open subset $M_* \subset M$  and a positive number $T$ such that for $x\in M_*$
\be
\phi_{t_1}(x) = \phi_{t_2}(x) \quad \Longleftrightarrow \quad (t_1 - t_2)\in T \bZ,
\label{23}\ee
where $\bZ$ denotes the set of integers.
It follows by continuity of the map $\phi$ \eqref{22} that $\phi_{t+T}(x) = \phi_t(x)$ holds for
all $x\in M\setminus M_*$ as well.
The complement of $M_*$ can contain equilibrium points and periodic orbits with period $T/n$ for some positive integers $n\neq 1$.
Strongly  isochronous Hamiltonians are also known as periodic Hamiltonians.
For any such Hamiltonian with basic period $T>0$,
the following formula defines a smooth, effective action of the Lie group $U(1) = \bR/(2\pi \bZ)$ on $M$:
\be
A: U(1) \times M \to M, \quad A(e^{\ri \tau}, x) := \phi_{\tau T/(2\pi)}(x).
\label{24}\ee
In what follows, we also use $A_{e^{\ri \tau}}(x):= A(e^{\ri \tau}, x)$.
The submanifold $M_*$ consists of those points whose $U(1)$ isotropy group is trivial.
If $H$ is strongly isochronous, then the commutant of $H$ in $C^\infty(M)$ is given by ring of $U(1)$ invariant
functions, denoted $C^\infty(M)^{U(1)}$.
In consequence of a general result presented below, this ring of invariants has functional dimension $\dim(M) -1$.

Now, suppose that we have a $C^\infty$ action of a compact Lie group $K$ on a connected manifold $M$.
Let $L$ be the isotropy subgroup, $L=K_{x_0}<K$, of some point $x_0\in M$, and
denote by $(L)$ the collection of the subgroups of $K$ that are conjugate to $L$.
Then,
 \be
 M_{(L)}:= \{ x\in M \mid K_x \in (L) \}
 \label{25}\ee
is a submanifold of $M$, called the $(L)$ orbit type submanifold.
 The orbit $K\cdot x$ is a closed submanifold of $M$,  diffeomorphic to $K/L$ for every $x\in M_{(L)}$.
According to a fundamental result \cite{DK,Mic}, there exists a unique conjugacy class of isotropy  subgroups $(L)$ of $K$ such that
$K_x$ contains an element of $(L)$ for every $x\in M$.
The corresponding subset $M_{(L)}$ is \emph{dense and open} in $M$, and is called the  \emph{submanifold of principal orbit type},
henceforth denoted $M_*^K$.
The orbits $K\cdot x$ for $x\in M_*$ are  the principal orbits.
Informally speaking, $M_*^K$ is the subset of elements having the smallest possible isotropy groups.
The quotient space $M_*^K/K$ is connected if $M$ is connected, and $M_*^K$ itself is connected if $K$ is also connected.

If the group $K$ is Abelian, then any conjugacy class of subgroups is just a single group. For an \emph{effective}
action of a compact  Abelian Lie group on a connected manifold, the principal isotropy group is trivial, since for an effective action
only the unit element acts as the identity transformation.

\begin{lem}\label{lem:1}
Consider the ring of invariant functions $C^\infty(M)^K$  with respect to a smooth action
of a compact Lie group $K$  on a connected smooth manifold $M$.
Let $s$ be the co-dimension in $M$ of a principal  orbit $K\cdot x \subset M_*^K$.
Then,
the exterior derivatives of the elements of $C^\infty(M)^K$ form an $s$-dimensional subspace of $T_x^*M$
at every $x\in M_*^K$.
\end{lem}

A proof of the this well known result can be found, for example, in \cite{WGMPnew}.
It implies that $C^\infty(M)^K$ has a well-defined functional dimension, $\mathrm{ddim}(C^\infty(M)^K)$, given by
\be
\mathrm{ddim}(C^\infty(M)^K) =  s \equiv \dim(M) - (\dim(K) - \dim(K_x)), \qquad \forall x \in M_*^K.
\label{26}\ee

The above-described preliminaries lead to the claim stated in the introduction.

\begin{thm}\label{thm:22}
 All strongly isochronous Hamiltonians are maximally superintegrable.
\end{thm}
\begin{proof}
If  $(M,\br{-,-},H)$ is a strongly isochronous Hamiltonian system, then the commutant of $H$ in $C^\infty(M)$ is
given by the invariants $C^\infty(M)^{U(1)}$ for the associated $U(1)$ action \eqref{24}.
In this case the principal isotropy group is trivial.
Thus, by \eqref{26}, the functional dimension of $C^\infty(M)^{U(1)}$ is
equal to $\dim(M)-1$.
\end{proof}

We do not claim any originality regarding theorem \ref{thm:22}, but thought it worthwhile to present
it together with the underlying group theoretic background.

\section{Reductions of strongly isochronous Hamiltonian systems}
\label{S:3}

Below, we formulate  simple sufficient conditions for guaranteeing that reductions of
strongly isochronous systems remain strongly isochronous.

\subsection{Poisson reduction}

Let $(M,\br{-,-},H)$  be a strongly isochronous system on a connected Poisson manifold $M$
that admits  a connected, compact Lie group $K$ as symmetry group.
By this we mean that we have a smooth action
\be
A^K: K \times M \to M
\label{31}\ee
such that the group elements act by Poisson diffeomorphisms that also preserve the Hamiltonian $H$.
Then, restricting to the submanifold of principal orbit type $M_*^K$ for the $K$ action, we obtain the reduced Hamiltonian system
\be
(M_*^\red, \br{-,-}_*^\red, H_*^\red).
\label{32}\ee
The reduced Poisson bracket on the \emph{smooth manifold}
\be
M_*^\red := M_*^K/K
\label{33}\ee
is defined by the natural identification
\be
C^\infty (M_*^K/K) \simeq C^\infty(M_*^K)^K.
\label{34}\ee
It should be noted that the isotropy group is constant along any integral curve of any $K$-invariant Hamiltonian.
In particular, the integral curves of $H$ preserve $M_*^K$, and they project onto the integral curves of the reduced Hamiltonian
$H_*^\red$ defined by the equality
\be
H \circ \iota = H_*^\red \circ \pi,
\label{35}\ee
where
\be
\iota: M_*^K \to M
\quad \hbox{and}\quad
\pi: M_*^K \to M_*^\red
\label{36}\ee
are the tautological embedding and the natural projection, respectively.

The $U(1)$ action (\ref{24}) generated by the Hamiltonian $H$ and the $K$ action commute, and they can be combined into
an action of the direct product group $U(1)\times K$ on $M$, which gives the Poisson diffeomorphisms
\be
A_{(e^{\ri \tau},g)} := A_{e^{\ri \tau}} \circ A_g^K, \qquad \forall (e^{\ri \tau}, g) \in U(1) \times K.
\label{37}\ee
Now, we make the assumption that the principal orbit type submanifolds for the $U(1) \times K$ and the $K$ actions,
denoted $M_*^{U(1) \times K}$  and $M_*^K$,
are related in the following way:
\be
M_*^{U(1) \times K} = \{ x \in M^K_* \mid (U(1) \times K)_x = \{1\} \times K_x \},
\label{38}\ee
where $1\in U(1)$ is the unit element.
Since $M$ is connected, and the isotropy group $(U(1) \times K)_x$ contains $\{ 1\} \times K_x$,
the assumption is equivalent to the existence of a single point $x\in M_*^K$
satisfying
\be
(U(1) \times K)_x =   \{ 1\} \times K_x.
\label{310}\ee
This guarantees that the conjugacy class of principal isotropy groups for the action (\ref{37}) is
 $\{1\} \times (K_x)$, where $K_x$ is a principal isotropy group for the $K$ action.

 The condition (\ref{38}) is easily seen to be equivalent to the existence of a dense open subset of $M$ whose points
 $x$ have the property
 \be
 U(1) \cdot x \cap K \cdot x = \{x\},
 \ee
 i.e., the $U(1)$ orbit and the $K$ orbit through $x$ intersect only in $x$.

\begin{prop}\label{prop:31}
Let $(M,\br{-,-},H)$ be a $K$-invariant strongly isochronous Hamiltonian system    on the connected manifold $M$,
and suppose that the $U(1)$ action (\ref{24}) generated by the Hamiltonian $H$ and the action of the connected, compact
symmetry group $K$ verify the condition  \eqref{38}.
Then, the reduced system \eqref{32}  on the space of principal $K$ orbits $M_*^\red$  \eqref{33} is strongly isochronous,
with the same basic period  as the original system.
\end{prop}
\begin{proof}
Since  $M_*^{U(1) \times K} \subset M$ and $M_*^K \subset M$ are dense open subsets, it follows from the assumption \eqref{38} that
 $M_*^{U(1) \times K} \subset M_*^K$ is dense open. Consequently,
\be
M_*^{U(1) \times K}/K  \subset M_*^K/K \equiv M_*^\red
\label{311}\ee
is also a dense open subset.

Next, observe that the Hamiltonian $U(1)$ action descends to the $U(1)$ action
\be
A^\red: U(1) \times M_*^\red \to M_*^\red
\label{312}\ee
 given (using $\pi$ in \eqref{36})  by
\be
A_{e^{\ri \tau}}^\red (\pi(x)) = \pi(A_{e^{\ri \tau}} (x)),
\qquad \forall e^{\ri \tau} \in U(1),\, x\in M_*^K.
\label{313}\ee
The definition of the Poisson reduction ensures that this reduced $U(1)$ action is generated by
the flow of the reduced Hamiltonian system \eqref{32}.
Thus, it is enough to show that the reduced $U(1)$ action is free on the subset $M_*^{U(1) \times K}/K$ \eqref{311}.
To see this, suppose that
\be
A_{e^{\ri \tau}}^\red (\pi(x)) = \pi(x) \quad\hbox{for}\quad x \in M_*^{U(1) \times K}.
\label{314}\ee
This means that there exists $g \in K$ for which $A_{e^{\ri \tau}}(x) = A^K_g(x)$, or equivalently
\be
A_{(e^{\ri \tau}, g^{-1})}(x) = x.
\label{315}\ee
By \eqref{38}, we must have $e^{\ri \tau} = 1$.
Therefore, the principal isotropy group for the reduced Hamiltonian $U(1)$ action on $M_*^\red$ is trivial,
and this immediately implies the claim.
\end{proof}

\begin{rem}
The reduced $U(1)$ action \eqref{313} can be restricted on any symplectic leaf of $M_*^\red$ \eqref{33}.
It follows from  \eqref{38}  that
the principal isotropy group for the restricted $U(1)$ action is trivial whenever the leaf intersects  $\pi(M_*^{U(1)\times K})$.
This entails that the restriction of the reduced Hamiltonian to every such symplectic  leaf is strongly
isochronous with the same basic period as the original Hamiltonian $H\in C^\infty(M)$.
\end{rem}

\begin{rem}
The constants of motion of the reduced system \eqref{32} are given by
\be
C^\infty(M_*^\red)^{U(1)}  = C^\infty(M_*^K)^{U(1) \times K},
\label{316}\ee
and with the embedding $\iota$ in \eqref{36} we have
\be
\iota^* \left( C^\infty(M)^{U(1) \times K}\right) \subset C^\infty(M_*^K)^{U(1) \times K}.
\label{317}\ee
This is in general a proper subset. By using  Lemma \ref{lem:1} and the relation \eqref{38}, one can show that the corresponding rings of functions
on $M_*^\red$ have the same functional dimension. The  elements  of $C^\infty(M)^{U(1) \times K}$
enjoy the advantage that they yield smooth functions on the full reduced phase space $M/K$.
\end{rem}

One may also study the reduced system on the full, stratified reduced phase space $M/K$, but this is outside the scope of the present paper.
For a recent work in this direction, see \cite{CJRX}

\subsection{Marsden--Weinstein Hamiltonian reduction}

Let$(M,\br{-,-},H)$  be a strongly isochronous system  with a connected, compact Lie group $K$ as symmetry group, as before, but now also assume
that the Poisson structure comes from a symplectic form, $\Omega$, and the $K$ action on $M$ is Hamiltonian with equivariant moment map
\be
J: M\to \cK^*.
\label{318}\ee
The dual $\cK^*$ of the  Lie algebra $\cK$ of $K$ carries the coadjoint action,
and for any $\mu\in J(M)$ the coadjoint isotropy group $K_\mu < K$ of $\mu$ acts on $J^{-1}(\mu)$.
Let us suppose that  $J^{-1}(\mu)$ is an embedded submanifold of $M$,
and that there is only a single orbit type
for the $K_\mu$ action on $J^{-1}(\mu)$.
Then, the Marsden--Weinstein--Meyer
reduction theorem \cite{OR}
 yields the reduced Hamiltonian system
\be
(M_\mu, \Omega_\mu, H_\mu)
\label{319}\ee
characterized by the following properties.\footnote{Usually it is required that the action of $K_\mu$ is free on $J^{-1}(\mu)$,
which guarantees that $\mu$ is a regular value for $J$, and thus $J^{-1}(\mu)$ is an embedded submanifold.}
First,  $M_\mu= J^{-1}(\mu)/K_\mu$ is a smooth manifold in such a way
that  the natural projection $\pi_\mu: J^{-1}(\mu) \to M_\mu$ is a smooth submersion.
Second, with the tautological inclusion $\iota_\mu: J^{-1}(\mu) \to M$, the reduced symplectic form $\Omega_\mu$  and reduced Hamiltonian $H_\mu$
obey the relations
\be
\pi_\mu^*(\Omega_\mu) = \iota_\mu^*(\Omega)
\quad\hbox{and}\quad
\pi_\mu^* (H_\mu) = \iota_\mu^*(H).
\label{320}\ee
In this case, the  integral curves of the evolution equation associated with the reduced system $(M_\mu, \Omega_\mu, H_\mu)$ are of the form
$\pi_\mu(x(t))$, where $x(t)$ is an integral  curve of the system $(M,\Omega,H)$ lying in $J^{-1}(\mu)$.

\begin{prop}\label{prop:34}
Consider the  action  of $U(1) \times K_\mu$ on $J^{-1}(\mu)$, which arises from the pertinent restriction of the action \eqref{37}.
Maintaining the above assumptions concerning $J^{-1}(\mu)$ and the $K_\mu$ action,
suppose  that every connected component of $J^{-1}(\mu)$ contains a point $x$ whose isotropy group is
\be
(U(1) \times K_\mu)_x = \{ 1\} \times (K_\mu)_x,
\label{321}\ee
where $(K_\mu)_x$ is the isotropy group of $x$ with respect to the $K_\mu$ action.
Then, the reduced system \eqref{319} is strongly isochronous, with the same basic period as the original system.
\end{prop}
\begin{proof}
This follows essentially by the same arguments that we applied in the preceding subsection.
The assumption (\ref{321}) implies that  the principal isotropy type submanifold with respect to the action of $U(1)\times K_\mu$ on $J^{-1}(\mu)$ is
\be
J^{-1}(\mu)_*^{U(1) \times K_\mu} = \{ x \in J^{-1}(\mu) \mid (U(1) \times K_\mu)_x = \{ 1\} \times (K_\mu)_x.
\label{322}\ee
The $U(1)$ action descends to $M_\mu$, and the principal isotropy type submanifold for
the resulting  $U(1)$ action is the $\pi_\mu$ projection of the set \eqref{322}.
On this dense open subset of $M_\mu$, the projected $U(1)$ action is free, which means that the
Hamiltonian flow of $H_\mu$ has
the same basic period as does the flow of the original Hamiltonian on $M_*$.
\end{proof}

\begin{rem}
The reduced system \eqref{319} is  not sensitive to the behavior of the original Hamiltonian flow outside $J^{-1}(\mu)$.
For its strongly isochronous nature, it is sufficient  to require that
the solutions of the original Hamiltonian system lying in $J^{-1}(\mu)$ are periodic with the same period over a dense open subset of $J^{-1}(\mu)$,
and \eqref{321} holds for the corresponding $U(1)$ action on $J^{-1}(\mu)$.
\end{rem}

\section{Confining spin Calogero--Moser models from  simple Lie algebras}
\label{S:4}

We here use Poisson reduction for constructing Lie algebraic generalizations of the rational Calogero--Moser model with harmonic term
that involve collective spin degrees of freedom.
This is a variant of the construction of spin Calogero--Moser type models presented in \cite{AKLM,Hoch}, without the harmonic term.
A similar construction of trigonometric models can be found in \cite{Re1}.

\subsection{Definitions and abstract characterization of the reduced system}
Let $\fg$ be a complex simple Lie algebra and denote $\langle -,-\rangle$ its Killing form
(possibly multiplied by a positive constant for convenience).
Regarding $\fg$ as a \emph{real} vector space, we can decompose it as
\be
\fg = \cK + \cP\quad\hbox {with}\quad \cP:= \ri \cK,
\label{41}\ee
where $\cK$ is  a maximal compact subalgebra.  By using the restriction of the Killing form, which  is positive definite on $\cP$,
we identify the real vector space $\cP$ with its dual space.
Then, we consider the phase space
\be
M:= T^*\cP \simeq \cP \times \cP = \{ (X,Y)\mid X,Y\in \cP\}.
\label{42}\ee
The components of $X$ and $Y$  are coordinates on $M$, and the canonical symplecic form
can be written as
\be
\Omega= \langle dY \stackrel{\wedge}{,} dX \rangle.
\label{43}\ee
A natural strongly isochronous system on $M$ is defined by the  Hamiltonian
\be
H(X,Y) = \frac{1}{2} \langle Y, Y \rangle +  \frac{1}{2} \omega^2 \langle X, X \rangle,
\label{44}\ee
 representing an isotropic harmonic  oscillator.
It is useful to collect $X$ and $Y$ in the $\fg$-valued variable
\be
Z:= \omega X - \ri  Y,
\label{45}\ee
since by using $Z$ the Hamiltonian evolution equation takes the form $\dot{Z} = \ri \omega Z$.
The integral curves are
\be
\phi_t(Z)= e^{\ri \omega t} Z, \qquad \phi_0(Z) =Z,
\label{46}\ee
and the corresponding Hamiltonian $U(1)$ action \eqref{24} reads
\be
A_{e^{\ri \tau}}(Z) =  e^{\ri \tau} Z.
\label{47}\ee
This is a free action if we remove the equilibrium point $Z=0$.

Let $K$ be a connected and simply connected compact Lie group whose Lie algebra is $\cK$.
The group $K$ acts on $\cK$ by the adjoint action, which also induces actions on $\cP = \ri \cK$ and on $M$.
We denote the adjoint action by conjugation,  and thus for $g\in K$  have
\be
A_g^K: (X,Y) \mapsto (g X g^{-1}, g Y g^{-1})
\quad \Leftrightarrow \quad A_g^K: Z \mapsto g Z g^{-1}.
\label{49+}\ee
This is a Hamiltonian action,  generated by the  equivariant moment map
\be
J(X,Y) = [X,Y],
\label{49++}\ee
where $\cK^*$ is identified  with $\cK$ with the aid of the Killing form.

\begin{lem}\label{lem:41}
The principal isotropy group for the action \eqref{47} of $K$ on $M$ \eqref{42} is the center $\cZ(K)$ of the group $K$.
\end{lem}
\begin{proof}
It is enough to exhibit a single element $Z$ \eqref{45} whose isotropy group is $\cZ(K)$.
To do this, consider two  maximal Abelian subalgebras $\cT$ and $\cT'$ of $\cK$ for which the
corresponding maximal tori $\bT$ and $\bT'$ of $K$ satisfy
\be
\bT \cap \bT' = \cZ(K).
\label{410}\ee
The existence of such subalgebras is well known, see e.g. \cite{Mein}. Let $ X\in \ri \cT$ and $ Y\in \ri \cT'$ be regular elements,
whose isotropy groups with respect to the  action of $K$ on $\cP$ are $\bT$ and $\bT'$, respectively.
Then, $Z = \omega X - \ri Y$ has the required property.
\end{proof}

\begin{thm}\label{thm:42}
Suppose that $\rank(\cK)>1$.
Then, the principal isotropy group for the action \eqref{37}  of $U(1) \times K$ on $M_*^K\subset M$ \eqref{42}  is $\{1\} \times \cZ(K)$.
Consequently, the Poisson reduction of the harmonic oscillator  \eqref{44} restricted on $M_*^K$ yields a strongly isochronous
(and thus maximally superintegrable)
Hamiltonian system on $M_*^\red = M_*^K/K$, with  period $T= 2\pi/\omega$.
\end{thm}
\begin{proof}
Since $M_*^K$ is connected, it is sufficient to find a single point  $(X,Y) \in M_*^K$, encoded by $Z$ \eqref{45}, for which
\be
(U(1) \times K)_Z = \{1\} \times \cZ(K).
\label{411}\ee
Now, choose $Z= \omega X - \ri Y $ used in the proof of Lemma \ref{lem:41} in such a way
that $\langle Z, Z\rangle \neq 0$ and suppose that $(e^{\ri \tau}, g) \in (U(1) \times K)_Z$.
Then, we have $e^{-\ri \tau}  Z = g Z g^{-1}$, and thus
\be
e^{-2 \ri \tau} \langle Z, Z\rangle = \langle e^{- \ri \tau}  Z, e^{- \ri \tau} Z \rangle =
\langle g Z g^{-1}, g Z g^{-1} \rangle = \langle Z, Z\rangle
\label{412}\ee
entails that $e^{-2 \ri \tau} = 1$. If $e^{-\ri \tau}$ was equal to $-1$, then we would get
\be
g X g^{-1} =- X \quad \hbox{and}\quad g Y g^{-1} = - Y.
\label{413}\ee
This in turn would entail that $g \in N_K(\cT)$ and $g \in N_K({\cT'})$.
If $\rank(\cK) > 1$, i.e., $\cK \neq {\mathfrak{su}}(2)$, then we can choose $X\in \ri \cT$ so that it is not an eigenvector with eigenvalue $-1$
for any element of the Weyl group $N_K(\cT)/\bT$. This choice guarantees that $e^{\ri \tau}=1$, and  therefore \eqref{411}
holds.
\end{proof}

\begin{rem}\label{rem:43}
Suppose that $\rank(\cK)=1$, i.e., $K=SU(2)$. If $(X,Y) \in M_*^K$, then
it can be shown that there exists $g\in K$ satisfying $g Z g^{-1} = - Z$.
This $g$ is unique up to
multiplication by -1, and it satisfies $g^{-1} = - g$.
Moreover, the isotropy  group of $(X,Y) \in M_*^{U(1)\times K}$ is
the group containing the four elements
\be
(1,\1_2),\quad (1,-\1_2), \quad (-1, g),\quad (-1,-g).
\label{414}\ee
One still obtains a strongly isochronous system on $M^K_*/K$, but its basic period is half that of
the unreduced system.
 We shall see that no spin degrees of freedom arise in this case.
\end{rem}

\subsection{Concrete description of the reduced system}

As a piece of preparation, we need the orthogonal direct sum decompositions
\be
\cK = \cT + \cT_\perp
\quad\hbox{and} \quad
\cP = \cA + \cA_\perp,
\label{415}\ee
where $\cT<\cK$ is a maximal Abelian subalgebra and $\cA = \ri \cT$.
The dense open subset $\cA_\reg\subset \cA$ contains the elements whose isotropy subgroup with respect to
the action of $K$ on $\cP$ is the maximal torus $\bT < K$ associated with $\cT$, and
$\cP_\reg \subset \cP$ is filled by the $K$ orbits intersecting $\cA_\reg$.
From now on, we focus on the \emph{dense open submanifold} ${\check M}_*^K \subset M_*^K$ defined by
\be
{\check M}_*^K:= \{(X,Y) \in M_*\mid X\in \cP_\reg\}.
\label{416}\ee
We recall from Lie theory that the connected components of $\cA_\reg$ are permuted by
the Weyl group
\be
W= N_K(\cT)/\bT,
\label{417}\ee
and fix a connected component $\cC \subset \cA_\reg$, alias an open Weyl chamber.
By using $\cC$, we introduce the `gauge slice'
\be
\cS:= \{ (q, Y) \in {\check M}_*^K \mid q \in \cC\}.
\label{418}\ee
Defining
\be
{\check M}_*^\red :={\check  M}_*^K /K,
\label{419}\ee
we obtain the identification
\be
{\check M}_*^\red \equiv \cS/\bT,
\label{420}\ee
since every $K$-orbit in ${\check M}_*^K$ intersects $\cS$ in an orbit of $\bT<K$.
At the level of rings of smooth functions, this  is equivalent to the isomorphism
\be
C^\infty( {\check M}_*^\red) \equiv C^\infty( {\check M}_*^K)^K \Longleftrightarrow C^\infty(\cS)^\bT.
\label{421}\ee
Thus, we can describe the observables of the (restricted) reduced system as $\bT$-invariant smooth functions
on the gauge slice $\cS$.

Next, we deal with the reduced Poisson bracket on  $C^\infty(\cS)^\bT$,  denoted $\{-,-\}_\cS$.
Its defining equality is
 \be
 \{ \iota^*(f), \iota^*(h) \}_\cS := \iota^* \{f,h\},
 \quad
 \forall f,h \in C^\infty({\check M}_*^K)^K,
\label{422} \ee
 where $\iota: \cS \to {\check M}_*^K$ is the tautological inclusion and $\{f,h\}$ is determined by the restriction
 of the symplectic form $\Omega$  (\ref{43})  on ${\check M}_*^K$.
For any $f \in C^\infty(M)$, denote $\nabla_1f$ and $\nabla_2f$  the $\cP$-valued partial gradients with
respect to the first and second variable. The Poisson bracket corresponding to $\Omega$ has the form
\be
\{ f,h\} = \langle \nabla_1 f, \nabla_2 h \rangle - \langle \nabla_2 f, \nabla_1 h \rangle,
\label{423}\ee
and of course the same formula is valid over the dense open submanifold ${\check M}^K_* \subset M$.
For a function $F\in C^\infty(\cS)$, $\nabla_2 F$ is $\cP$-valued, but $\nabla_1 F$ is $\cA$-valued.
To be clear, the latter is defined by the requirement
\be
\langle V, \nabla_1 F(q,Y) \rangle = \dt F(q+ t  V,Y), \qquad \forall V\in \cA,
\label{424}\ee
which makes sense since $(q+ t V,Y)$ belongs to $\cS$ \eqref{418} for small enough $t$.

For any $q\in \cA_\reg$, $\ad_q := [q,\,\cdot\,] \in \mathrm{End}(\fg)$ restricts to an invertible linear map of
$(\cT_\perp + \cA_\perp)$ onto itself, where we refer to the decompositions (\ref{41}) and (\ref{415}).
By definition, the linear map
\be
r(q): \fg  \to \fg
\label{425}\ee
operates on $(\cT_\perp + \cA_\perp)$ as the inverse of the pertinent restriction of $\ad_q$,
and $r(q)$ is set to be zero on $(\cT + \cA)$.
To display $r(q)$ explicitly, denote $\Delta= \{\alpha\}$ the set of roots of the complex Lie algebra $\fg$
with respect to the Cartan subalgebra $\cT^\bC$
 and
choose root vectors $E_\alpha\in \fg $ in such a way that
\be
E_{-\alpha} = (E_\alpha)^\dagger
\quad \hbox{and}\quad
 \langle E_\alpha, E_{-\alpha} \rangle = \frac{2}{\vert \alpha \vert^2}
\label{426}\ee
hold, where  dagger denotes complex conjugation on $\fg$ with respect to $\cK$.
Any element $\xi$ of $(\cT_\perp + \cA_\perp)$ can be expanded in the form
\be
\xi=\sum_{\alpha\in \Delta} \xi^\alpha E_\alpha,
\quad \xi^{-\alpha} = - (\xi^\alpha)^*
\quad \hbox{if}\quad \xi\in \cT_\perp
\quad\hbox{and}\quad \xi^{-\alpha} = (\xi^\alpha)^*
\quad\hbox{if}\quad \xi\in \cA_\perp.
\label{427}\ee
With this notation, we have
\be
r(q)( \xi) = \sum_{\alpha \in \Delta} \frac{\xi^\alpha}{\alpha(q)} E_\alpha .
\label{428}\ee
Observe that $r(q)$ is antisymmetric,
\be
\langle r(q) U, V \rangle = - \langle U, r(q) V \rangle, \qquad \forall U,V \in \fg,
\label{429}\ee
and it maps $\cT_\perp$ onto $\cA_\perp$  and $\cA_\perp$ onto $\cT_\perp$.

\begin{prop}\label{prop:44}
The reduced Poisson bracket (\ref{422}) on $C^\infty(\cS)^\bT$  obeys the explicit formula
\be
\{F,H\}_\cS(q,Y) = \langle \nabla_1 F, \nabla_2 H \rangle  - \langle \nabla_2 F, \nabla_1 H \rangle
+ \langle Y,  [r(q) \nabla_2 F, \nabla_2 H ] + [\nabla_2 F, r(q) \nabla_2 H] \rangle,
\label{430}\ee
with the derivatives taken at $(q,Y)\in \cS$ (\ref{418}).
\end{prop}

\begin{proof}
Referring to the decomposition (\ref{415}), notice that
only ${(\nabla_2 F)}_{\cA}$ and ${(\nabla_2 H)}_{\cA}$ appear in the first two terms of \eqref{430}, since $\nabla_1 F$ and $\nabla_1H$ belong to $\cA$.
Any $F\in C^\infty(\cS)^\bT$ is the restriction of a unique function $f\in C^\infty({\check M}_*^K)^K$, and
the $K$-invariance of $f$ implies that
\be
f(q,Y) = f(e^{t U} q e^{-tU}, e^{tU} Y e^{-tU}), \qquad \forall (q,Y)\in \cS,\,\, U\in \cK,\,\, t \in \bR.
\label{431}\ee
This leads to the identity
\be
[q, \nabla_1 f(q,Y)] + [Y, \nabla_2 f(q,Y)] =0,
\label{432}\ee
which permits us to express $(\nabla_1 f(q,Y))_{\cA_\perp}$ in terms of $\nabla_2 f(q,Y)$.
Taking into account the obvious relations
\be
\nabla_2 f(q,Y) = \nabla_2 F(q,Y)
\quad \hbox{and}\quad
(\nabla_1 f(q,Y))_\cA = \nabla_1 F(q,Y),
\label{433}\ee
we obtain
\be
\nabla_1 f(q,Y) = \nabla_1 F(q,Y) - r(q) [Y, \nabla_2 F(q,Y)].
\label{434}\ee
We expressed the derivatives of $f$ at $(q,Y)\in \cS$ in terms of the derivatives of
the restricted function $F= \iota^*(f)$.
We  substitute these expressions and their counterparts for $H\in C^\infty(\cS)^\bT$ into the definition  (\ref{422}), i.e., into
$\{ F,H\}_\cS(q,Y) = \{ f,h\}(q,Y)$.
It is then straightforward to derive (\ref{430}) from  (\ref{423}) utilizing the antisymnetry of $r(q)$.
   \end{proof}

The physical interpretation of the model will be transparent in terms of new variables.
Namely, we introduce new variables $(q,p,\xi)$ instead of $(q,Y)$ by means of the following
diffeomorphism:
\be
m: \cC \times \cA \times \cT_\perp \to \cC \times \cP,
\quad
m: (q,p, \xi) \mapsto (q, p - r(q) \xi),
\label{435}\ee
i.e., by parametrizing
$Y\in \cP$ as
\be
Y(q,p,\xi) = p - r(q)\xi.
\label{Ypar}\ee
This is a $\bT$-equivariant map if $\bT$ acts by conjugations on $Y$ as well as on $\xi\in \cT_\perp$.

For any real function $\cF \in C^\infty(\cC \times \cA \times \cT_\perp)$, we denote its partial gradients by
$\nabla_q \cF, \nabla_p \cF$, which are $\cA$-valued, and $\nabla_\xi \cF$, which is $\cT_\perp$-valued,
\be
\langle U, \nabla_\xi \cF(q,p,\xi)\rangle :=\dt \cF(q,p, \xi + t U),  \qquad
\forall U \in \cT_\perp.
\label{436}\ee
The reduced Poisson bracket in terms of the new variables is determined by the equality
 \be
 \{ F \circ m, H\circ m \}_\red := \{ F,H\}_\cS \circ m,
 \qquad
 \forall F,H \in C^\infty(\cS)^\bT.
 \label{437}\ee

\begin{prop}\label{prop:45}
In the variables $(q,p,\xi)$ introduced using the map $m$ (\ref{435}), the reduced Poisson bracket (\ref{437}) is given explicitly by
\be
\{\cF, \cH\}_\red(q,p,\xi) = \langle \nabla_q \cF, \nabla_p \cH \rangle - \langle \nabla_p \cF, \nabla_q \cH \rangle +
\langle \xi, [\nabla_\xi \cF, \nabla_\xi \cH] \rangle
\label{438}\ee
for all $\cF, \cH \in C^\infty(\cC\times \cA \times \cT_\perp)^\bT$, where the derivatives are evaluated at $(q,p,\xi)$.
The reduction  of the Hamiltonian  (\ref{44}) yields
\be
\cH_\red(q,p,\xi) := \frac{1}{2} \langle Y(q,p,\xi), Y(q,p,\xi)\rangle + \frac{1}{2} \omega^2 \langle q,q  \rangle,
\label{439}\ee
with $Y(q,p,\xi)$ in \eqref{Ypar}, and this
 can be spelled out as
\be
\cH_\red(q,p,\xi) = \frac{1}{2} \langle p, p\rangle +  \frac{1}{2} \omega^2 \langle q, q\rangle +
\sum_{\alpha \in \Delta_+} \frac{2}{\vert \alpha \vert^2} \frac{\vert \xi^\alpha \vert^2}{\alpha(q)^2},
\label{440}\ee
where the expansion $\xi = \sum_{\alpha \in \Delta} \xi^\alpha E_{\alpha}$ is used  (\ref{427}), and $\Delta_+$ is the set of positive roots.
\end{prop}

\begin{proof}
Suppose that $\cF = F \circ m$, which means that
\be
\cF(q,p,\xi) = F(q, p - r(q)\xi),
\qquad
\forall q,p,\xi \in \cC \times \cA \times \cT_\perp.
\label{441}\ee
A routine calculation leads to the following relation of the derivatives of the functions $\cF$ and $F$:
\be
\nabla_2 F(q, Y(q,p,\xi)) = \nabla_p \cF(q,p,\xi) + [q, \nabla_\xi \cF(q,p,\xi)]
\label{442}\ee
and
\be
\nabla_1 F(q, Y(q,p,\xi)) = \nabla_q \cF(q,p,\xi) + [ r(q)\xi, \nabla_\xi \cF]_\cA,
\label{443}\ee
where the last subscript refers to the decomposition $\cP = \cA + \cA_\perp$.
Then, the formula (\ref{430}) is converted into (\ref{438}) by
direct substitution of the above relations, and the similar relations for $\cH = H \circ m$,
into the defining equality
\be
\{ \cF, \cH\}_\red(q,p,\xi) = \{ F,H\}_\cS (q, p- r(q) \xi).
\label{444}\ee
 The required manipulations are straightforward, and are thus omitted.
\end{proof}

The reduced Hamiltonian (\ref{440}) describes $N ={ \mathrm{dim}}(\cA)$ `point particles' moving along a line in an external
harmonic potential,  interacting according  to the inverse square potential associated with the root system,
and also interacting  with the `collective spin variable' $\xi$.
One may extend the position variable $q$ from the fixed Weyl chamber  $\cC$ to $\cA_\reg$, which enlarges the residual gauge symmetry
from $\bT$ to the normalizer $N_K(\cA)$.
The third term of the Poisson bracket \eqref{438} can be recognized as the reduction of Lie--Poisson bracket on $\cK^* \simeq \cK$
that arises by setting to zero the moment map of the natural $\bT$-action on $\cK^*$
(which is just the component $\xi_\cT$ of $\xi \in \cK$).

By applying them to $\xi$, the invariant functions $C^\infty(\cK)^K$ give rise to Casimir functions of the reduced Poisson bracket.
In the rank 1 case $\xi^\alpha\neq 0$ can be made positive by the $\bT$ gauge transformations, and its value then depends only
on the quadratic Casimir  $\langle \xi, \xi \rangle$. Thus, in this case, the `spin' $\xi$ is not a proper dynamical variable.

Finally, it is worth noting that the map $m$ (\ref{435}) was not  guessed, but has its natural origin in the alternative treatment
of the Poisson reduction of $T^* \cP$ by means of the so-called shifting trick \cite{OR}.   This works by first considering
the extended phase space $T^*\cP \times \cK^*$, equipped with the natural diagonal action of $K$.
Then one sets the corresponding moment map, $J(X,Y,\xi) = [X,Y] + \xi$, to zero, and discovers the formula (\ref{Ypar})
of $Y(q,p,\xi)$ after gauge fixing $X$ to (a connected component of) $\cA_\reg$.
The formula \eqref{Ypar} also appears in the construction  of spin Calogero--Moser models based on dynamical
$r$-matrices \cite{LX}. Indeed,  $r(q)$ in (\ref{425}) is a restriction  on $\cA_\reg\subset \cT^\bC_\reg$ of the standard rational dynamical
$r$-matrix of $\fg$.

\section{Confining spin Calogero--Moser model of Gibbons--Hermsen type}
\label{S:5}

Now we wish to demonstrate the strongly isochronous nature of the standard spin Calogero--Moser model with harmonic term,
which on a dense open subset of its phase space is governed by the Hamiltonian
\be
\cH_{\mathrm{spin}}(q,p, \zeta) =\frac{1}{2}\sum_{i=1}^n (p_i^2 + \omega^2 q_i^2) +
\frac{1}{2} \sum_{ i\neq j} \frac{\vert \zeta_i \zeta_j^\dagger\vert^2}{(q_i - q_j)^2}\,.
\label{51}\ee
Here, $q_i$ and  $p_i$ $(i=1,\dots, n)$ are interpreted as positions and momenta of $n$ `point particles' moving on the real line.
Each particle carries an $\ell$-dimensional ($\ell>1$) complex spin vector, $\zeta_i$, which are collected in the rows of the matrix $\zeta\in \bC^{n\times \ell}$.
More precisely, the length of each spin vector $\zeta_i$ is fixed to the same positive value (see (\ref{514}) below), and these degrees of freedom matter up to the gauge transformations
\be
(Q,P,\zeta) \mapsto (g Q g^{-1}, g P g^{-1}, g\zeta), \qquad \forall g \in \fN,
\label{52}\ee
where
\be
Q:= \diag(q_1,\dots, q_n),
\quad P:= \diag(p_1,\dots, p_n),
\label{53}\ee
and
\be
\fN := N_{U(n)}(\cT)
\label{54}\ee
 is the normalizer of the standard  Cartan subalgebra  $\cT < \mathfrak{u}(n)$.
The full phase space of the model results from the Hamiltonian reduction summarized below, which goes back
to Gibbons and Hermsen \cite{GH}. (In \cite{GH} the  complex case was considered, actually without the harmonic term.)
If $\ell=1$, then the spin degrees of freedom disappear via the gauge transformations,
and in this case the strongly isochronous character of the Hamiltonian \eqref{51} was established by Adler \cite{Ad}.

Let us denote $\cP:= \ri \mathfrak{u}(n)$.
The starting point for the derivation of the model \eqref{51} is the phase space
\be
M:= \cP \times \cP \times \bC^{n\times \ell} = \{(X,Y, \zeta)\}
\label{55}\ee
equipped with the symplectic form $\Omega$ and the strongly isochronous Hamiltonian $H$:
\be
\Omega= \tr\left( dY \wedge dX -\ri d\zeta \wedge d\zeta^\dagger\right),\quad
H(X,Y,\zeta) = \frac{1}{2} \langle Y, Y \rangle +  \frac{1}{2} \omega^2 \langle X, X \rangle,
\label{56}\ee
where we use the trace form $\langle U,V\rangle := \tr(UV)$  on $\mathfrak{gl}(n,\bC) = \mathfrak{u}(n)^\bC$ and its relevant subspaces
(cf. equations  \eqref{43} and \eqref{44}).
Combining $X$ and $Y$ in the new variable $Z:=\omega X - \ri Y$, like in (\ref{45}), the Hamiltonian flow through the initial value $(Z,\zeta)$ gives
\be
\phi_t(Z,\zeta)= (e^{\ri \omega t} Z, \zeta).
\label{56+}\ee
This is periodic with period $T = \frac{2\pi}{\omega}$,
and thus  the definition \eqref{24} yields the $U(1)$ action
\be
A_{e^{\ri \tau}}(Z,\zeta) =  (e^{\ri \tau} Z,\zeta).
\label{56++}\ee

The system $(M,\Omega,H)$  is reduced utilizing the Hamiltonian $U(n)$ action
\be
A^{U(n)}_g: (X,Y, \zeta) \mapsto (gX g^{-1}, g Y g^{-1}, g\zeta), \quad \forall g\in U(n).
\label{57}\ee
The pertinent moment map is provided by
\be
J(X,Y,\xi) = [X,Y] + \ri \zeta \zeta^\dagger,
\label{58}\ee
where $\mathfrak{u}(n)^*$ is identified with $\mathfrak{u}(n)$  by means of the trace form.
The reduction is defined by constraining $J$  to a multiple of the unit matrix, $\ri c \1_n$ with
an arbitrarily chosen positive constant $c$.
It is known from the analysis  of the corresponding complex holomorphic case \cite{Wil} that $U(n)$ acts \emph{freely} on
 the `constraint surface'
\be
M^c:= J^{-1}(\ri c \1_n),
\label{59}\ee
 which is therefore
a closed, embedded submanifold of $M$. It is also known that $M^c$ is connected\footnote{The connectedness follows, for example, by combining  Theorem 2.2
and  case iii) of Theorem 4.2 in \cite{Knop}.}.
We shall apply proposition \ref{prop:34}  to show that the reduced system,  denoted
\be
(M_\red, \Omega_\red, H_\red) \quad \hbox{with}\quad M_\red := M^c/U(n),
\label{510}\ee
is strongly isochronous.

Let $\mathfrak{u}(n)_\reg\subset \mathfrak{u}(n)$ be the subset of regular elements having $n$ distinct eigenvalues.
Denote  $\cP_\reg:= \ri \mathfrak{u}(n)_\reg$ and
$\cT_\reg := \cT \cap \mathfrak{u}(n)_\reg$. Then, let us introduce
\be
M^\reg:= \{ (X,Y, \zeta) \in M\mid X \in \cP_\reg \}.
\label{511}\ee
This is a dense open subset of $M$.
It turns out that $M^\reg \cap M^c$ is not empty.
Consequently,
\be
M^c \cap M^\reg \subset M^c
\quad \hbox{and}\quad M_\red^\reg:= (M^{c}\cap M^{\reg})/U(n) \subset M_\red
\label{512}\ee
are dense open subsets of $M^c$ and $M_\red$, respectively.
To see this, we may use that the connected manifold $M^c$ inherits  a real analytic structure
from $M$,
and the analytic function given by the discriminant of $X$ is not identically zero on $M^c$.

Observe that
every $U(n)$ orbit in $M^c \cap M^\reg$ has representatives in the `gauge slice'
\be
S:= \{ (Q,Y, \zeta) \in M^c \cap M^\reg \mid  Q \in \ri \cT_\reg\}.
\label{513}\ee
When evaluated on $S$,  the diagonal entries of the moment map condition
$J(Q,Y, \zeta) = \ri c \1_n$  imply that
\be
\vert \zeta_j \vert^2 := \zeta_j \zeta_j^\dagger = c, \qquad \forall j=1,\dots,n.
\label{514}\ee
It also follows that $Y$ takes the form
\be
Y_{jk} = p_j \delta_{jk}  -\ri (1- \delta_{jk})\frac{\zeta_j \zeta_k^\dagger} {q_j - q_k}, \qquad  1\leq j,k \leq n,
\label{515}\ee
where  $Q\in \ri \cT_\reg$ and $P\in \ri \cT$ \eqref{53} are arbitrary.
The $U(n)$ transformations  mapping $S$ to itself are given by \eqref{52}, and this leads to the
the identification
\be
M_\reg^\red =(M^c\cap M^\reg)/U(n) \equiv  S / \fN.
\label{516}\ee
Moreover, using the permutation matrices that are contained in $\fN$ \eqref{54}, we can bring $Q$ into the open Weyl chamber
\be
\cC:= \{ Q \in \ri \cT \mid q_1 > q_2 > \cdots > q_n\}.
\label{517}\ee
Then, the  residual gauge transformations belong to the maximal torus $\bT < U(n)$, and
they multiply each $\zeta_i$  by an arbitrary element of $U(1)$. Since the length of each $\zeta_i$ is fixed,
we finally obtain the identification
\be
M_\red^\reg  = \cC \times \ri\cT \times (\bC \bP^{\ell-1} \times \cdots \times \bC \bP^{\ell-1}).
\label{518}\ee
Regarding $\cC \times \ri\cT$ as a model of $T^* \cC$, the product structure \eqref{518} holds in the symplectic sense, where
all the $n$ copies of the complex projective space $\bC \bP^{\ell-1}$ carry a multiple of the Fubini--Study symplectic form.
Here, we use that $\bC \bP^{\ell-1}$ is the symplectic reduction of $\bC^\ell$ defined by a moment map constraint
of the form \eqref{514} for the natural $U(1)$ action.
The Hamiltonian $\cH_{\mathrm{spin}}$ \eqref{51} results by restricting $H$ \eqref{56} on the gauge slice $S$ \eqref{513}, and it represents
the reduced Hamiltonian $H_\red$ on the dense open subset  $M^\reg_\red \subset M_\red$.

\begin{thm}\label{thm:51}
The  confining  spin Calogero--Moser model given by reduced system \eqref{510} is strongly isochronous,
 with period $T=2\pi/\omega$, for all values of the parameters $\omega>0$ and $c>0$.
\end{thm}
\begin{proof}
Using  $Z=\omega X - \ri Y$ and $\zeta$ as our variables,  we  consider the combined action of $U(1) \times U(n)$ on $M^c$
whereby $(e^{\ri \tau}, g)\in U(1) \times U(n)$ sends $(Z,\zeta)$ to $(e^{\ri \tau} g Z g^{-1}, g\zeta)$.
By Proposition \ref{prop:34} and the fact that $M^c$ is connected,  it is sufficient to exhibit a single element
$(Z,\zeta)\in M^c$ having trivial isotropy group,
\be
(U(1) \times U(n))_{(Z,\zeta)} = \{ 1\} \times \{\1_n\}.
\label{519}\ee
Now we form $Z:= \omega Q - \ri Y $ from an element of $S$ \eqref{513} for which
 $Q \in \cC$ \eqref{517} satisfying also $q_n>0$ and $P = \1_n$ in (\ref{515}).  Then $\langle Z, Z \rangle \neq 0$,
and  similarly to (\ref{412})  we get that
$e^{2\ri \tau} =1$ for any $(e^{\ri \tau}, g)$ from the isotropy group. If $e^{\ri \tau}=-1$, then $g Z g^{-1} = - Z$  and $g \zeta = \zeta$ must hold.
But this would imply that $g Q g^{-1} = - Q$, which is impossible since $gQ g^{-1}$ has the same eigenvalues as $Q$,
and all eigenvalues of our chosen $Q$ are positive.
In conclusion, we exhibited a point $(Z,\zeta) \in M^c$ for which
\be
(U(1) \times U(n))_{(Z, \zeta)} = \{ 1\} \times U(n)_{(Z,\zeta)} = \{1\} \times \{\1_n\}.
\label{521}\ee
In the last step,  we utilized that $U(n)_{(Z,\zeta)}= \{\1_n\}$, because $U(n)$ acts freely on $M^c$.
\end{proof}

\begin{rem}
It is worth pointing out that the spin Calogero--Moser model \eqref{51} admits a family of non-degenerate, compatible Poisson structures.
These can be obtained by replacing the unreduced symplectic form $\Omega$ \eqref{55} with
\be
\Omega_z := \Omega+ \sum_{1\leq \alpha < \beta \leq \ell} z_{\alpha \beta} d \vert v_\alpha\vert^2 \wedge d \vert v_\beta \vert^2,
\label{523}\ee
where the $z_{\alpha \beta}\in \bR$ are arbitrary parameters and $v_\alpha\in \bC^n$  $(\alpha = 1,\dots, \ell)$ denotes the $\ell$-th column
of the matrix $\zeta \in \bC^{n\times \ell}$.  It follows from  Remark 3.10 in \cite{FF} that $\Omega_z$ is symplectic for
each choice of  $z_{\alpha \beta}$, and with respect to this symplectic structure \eqref{57} still gives a Hamiltonian action on
$M$, with the same moment map \eqref{58}.
(In \cite{FF} $T^* U(n) \times \bC^{n\times \ell}$ was considered, and further arbitrary parameters $x_\alpha\in \bR$
were also introduced, which are here set to $1$.)
The unreduced Hamiltonian $H$  \eqref{56} generates the same dynamics through any of these Poisson structures on $M$,
which induce compatible, non-degenerate Poisson structures for the reduced system on $M_\red = M^c/U(n)$.
\end{rem}

\section{$B_n$ type generalization of the Gibbons--Hermsen model}
\label{S:6}

We here derive a family of maximally superintegrable  models associated with the $B_n$ root systems.
These are real forms of models considered in  \cite{FG} in the complex case, where superintegrability was
shown for generic values of the circular frequency $\omega$.

We start with the reductive Lie group $G_0:= U(n,n)$, which is the group $2n$ by $2n$ complex matrices $g$ satisfying
\be
g^\dagger I_{n,n} g = I_{n,n}
\ee
where  $I_{n,n}:= \diag(\1_n, - \1_n)$. The Lie algebra $\fG0:=  u(n,n)$ admits the decomposition
\be
\fG0= \cK_0 + \cP_0,
\ee
where $\cK_0$ and $\cP_0$ contain, respectively, the anti-Hermitian and the Hermitian matrices belonging to $\fG0$.
This is an orthogonal  decomposition with respect to the \emph{real} bilinear form
$\langle U, V \rangle := \tr(UV)$, $\forall U,V \in\fG0$.
Concretely, $\cK_0 = \mathfrak{u}(n) \oplus \mathfrak{u}(n)$ consists of block diagonal matrices, and $\cP_0$
consists of the block off-diagonal matrices of the form
\be
X= \begin{pmatrix} 0 & x\\ x^\dagger & 0
\end{pmatrix}, \qquad \forall x \in \bC^{n\times n}.
\label{A3}\ee
We identify the dual space of $\cP_0$ with itself by the trace form, and consider the unreduced phase space
\be
M:= T^* \cP_0 \times \bC^{n \times \ell_1} \times \bC^{n\times \ell_2} = \cP_0 \times \cP_0 \times  \bC^{n \times \ell_1} \times \bC^{n\times \ell_2} = \{ (X,Y,\zeta,\eta)\},
\ee
where $\ell_1$ and $\ell_2$ are arbitrary non-negative integers, so that $\ell_1 + \ell_2 >0$.
Without loss of generality, we shall assume that $\ell_1 >0$; if $\ell_2 =0$, then the variable $\eta$ is simply not present.

We equip the phase space $M$ with the symplectic form
\be
\Omega = \tr (dY \wedge dX - \ri d\zeta \wedge d \zeta^\dagger - \ri d \eta \wedge d\eta^\dagger)
\label{A5}\ee
and the harmonic oscillator Hamiltonian
\be
H(X,Y, \zeta, \eta):= \frac{1}{2} \tr( Y^2) + \frac{1}{2} \omega^2 \tr(X^2).
\label{HB1}\ee
Then, we consider Hamiltonian reduction under the maximal compact subgroup $K:= U(n) \times U(n)$ of $G_0$.
The element $g:= (g_1,g_2) \in K$, best thought of as a block diagonal matrix of size $2n$ by $2n$,
acts on $M$  by the symplectomorphism
\be
A_g: (X,Y, \zeta, \eta) \mapsto (gX g^{-1}, g Y g^{-1}, g_1 \zeta, g_2\eta).
\label{A7}\ee
By parametrizing $X$ with $x$ and similarly $Y$ with $y$ according to \eqref{A3},
the action becomes
\be
A_g: (x,y, \zeta, \eta) \mapsto (g_1 x g_2^{-1}, g_1 y g_2^{-1}, g_1 \zeta, g_2 \eta).
\ee
Identifying $u(n) \oplus u(n)$ with its dual space via the trace form, the corresponding moment map is given by
\be
J(X,Y, \zeta, \eta) = [X,Y] + \begin{pmatrix} \ri \zeta \zeta^\dagger & 0 \\0&  \ri \eta \eta^\dagger
\end{pmatrix} = \begin{pmatrix}  x y^\dagger - y x^\dagger + \ri \zeta \zeta^\dagger & 0 \\0& x^\dagger  y -
 y^\dagger x +  \ri \eta \eta^\dagger \end{pmatrix}.
\label{A8}\ee
We fix the moment map to the value
\be
\mu_0 := \diag( \ri c_1 \1_n, \ri c_2 \1_n),
\label{A9}\ee
where $c_1$ and $c_2$ are real constants, satisfying $c:= c_1 + c_2 >0$.
One can see from comparison with the complex case (equation (3.6) in \cite{FG} for $m=2$, with $\lambda_0 = c_1$, $\lambda_1 = c_2$)
that $K$ acts freely on $J^{-1}(\mu_0)$ if the regularity conditions
\be
c_1 c_2 \neq 0
\quad\hbox{and}\quad c_1/c_2 \neq - 1 + 1/k, \quad \forall k \in \bZ\setminus \{0,1\}
\label{regcond}\ee
are satisfied, which we henceforth assume. It follows from results of \cite{Knop} that $J^{-1}(\mu_0)$ is connected.
Consequently, $M_\red:= J^{-1}(\mu_0)/K$ is a smooth (even real-analytic), connected symplectic manifold.

In order to obtain a physical interpretation of the reduced Hamiltonian, we apply the singular value decomposition to  bring the matrix $x$
to diagonal form with real entries. Moreover, we restrict to the dense open subset of $J^{-1}(\mu_0)$
containing the $K$-orbits that intersect the `gauge slice', denoted $S_0$,
on which $x$ in \eqref{A3} has the form
\be
x=\diag(q_1, \dots, q_n)/\sqrt{2}
\quad\hbox{with}\quad  q_i \in \bR,\,\,  q_i\neq 0,\,\,
 (q_i \pm q_j)\neq 0\,\,\, \forall i\neq j.
\label{S0}\ee
The $\sqrt{2}$ is included for convenience.
We could arrange the $q_i$ to be all positive and ordered as well, but here we do not do so in order to bring out a residual Weyl group symmetry.

The next proposition can be verified by straightforward calculation, which we omit.

  \begin{prop}
  In the diagonal gauge $S_0$, defined by the condition \eqref{S0},   the moment map  constraint implies
  the relation
  \be
  \zeta_j \zeta_j^\dagger + \eta_j \eta_j^\dagger = c, \qquad \forall j=1,\dots, n,
  \label{c+}\ee
  which explains why $c= c_1 + c_2$ must be positive,
   and $y\in \bC^{n\times n}$ can be expressed as
  \be
  y_{jj} = \frac{1}{\sqrt{2}} \left(p_j + \ri \frac{c_2 - \eta_j \eta_j^\dagger}{q_j}\right),
  \ee
  where $p_j\in \bR$ is arbitrary, and
  \be
  y_{jk} =- \frac{\ri}{\sqrt{2}} \frac{\zeta_j \zeta_k^\dagger + \eta_j \eta_k^\dagger}{q_j - q_k}
  + \frac{\ri}{\sqrt{2}} \frac{\zeta_j \zeta_k^\dagger - \eta_j \eta_k^\dagger}{q_j + q_k},
  \quad \hbox{for}\quad j\neq k.
  \ee
  On the gauge slice $S_0$ the Hamiltonian $H$ \eqref{HB1} takes the form
  \be
 H_{S_0}=
\frac{1}{2} \sum_{j=1}^n \left[p_j^2 + \omega^2 q_j^2 + \frac{(c_2 - \eta_j \eta_j^\dagger)^2}{ q_j^2}\right]  +
\frac{1}{2} \sum_{j\neq k}\left[ \frac{\vert \zeta_j \zeta_k^\dagger + \eta_j \eta_k^\dagger\vert^2}{(q_j - q_k)^2}
 +\frac{\vert \zeta_j \zeta_k^\dagger - \eta_j \eta_k^\dagger\vert^2}{(q_j + q_k)^2}\right],
  \label{HB2}\ee
  and the symplectic form $\Omega$ \eqref{A5} gives
  \be
  \Omega_{S_0} =  \sum_{j=1}^n dp_j \wedge dq_j -\ri \sum_{j=1}^n (d\zeta_j \wedge d \zeta_j^\dagger + d \eta_j \wedge d\eta_j^\dagger).
  \label{OmS0}\ee
  \end{prop}

The gauge slice $S_0$ engenders the dense open subset $S_0/N_K(\cA_0)$ of $J^{-1}(\mu_0)/K$, where
$N_K(\cA_0)<K$ is the normalizer of the maximal Abelian subspace $\cA_0$ of $\cP_0$ containing the matrices  $X$ \eqref{A3}
with arbitrary real diagonal $x$.  It is well known that $N_K(\cA_0)$ is generated by the elements of the form
\be
(g_1, g_2) = (R z, R z)
\quad
\hbox{and}\quad
(g_1, g_2) = (E_{ii}z, - E_{ii}z),
\label{resK}\ee
where $R\in U(n)$ is a permutation matrix, including the identity, $z$ is a diagonal element of $U(n)$, and $E_{ii}$ denotes
the $n\times n$ matrix verifying $(E_{ii})_{jk} = \delta_{ij} \delta_{ik}$.
The corresponding gauge transformations act as arbitrary permutations and sign changes
of the variables $q_1,\dots, q_n$ (and $p_1,\dots, p_n)$), which constitute the Weyl group of $B_n$ type.
The `residual gauge transformations' \eqref{resK} also act on the spin variables $\zeta$ and $\eta$,
according to \eqref{A7}. Of course, the Hamiltonian $H_{S_0}$ \eqref{HB2} and $2$-form $\Omega_{S_0}$ \eqref{OmS0}
are invariant with respect to these transformations.

\begin{thm}
If both $\ell_1$ and $\ell_2$ are non-zero and the regularity conditions \eqref{regcond} hold, then
the $B_n$ type confining spin Calogero--Moser Hamiltonian
that descends from $H$ \eqref{HB1} is strongly isochronous on the reduced phase space $J^{-1}(\mu_0)/K$,
with period $T= 2\pi/\omega$.
\end{thm}
\begin{proof}
Following the proof of Theorem \ref{thm:51}, it is sufficient to exhibit a single point of $M^c = J^{-1}(\mu_0)$ whose
isotropy group with respect to the $U(1) \times K$ action is trivial.
To do so, take $(X,Y, \zeta, \eta)\in M^c$ satisfying the gauge fixing condition \eqref{S0} and also the additional requirements
\be
\tr (\omega X - \ri Y)^2 \neq 0\quad \hbox{and}\quad
\zeta_1^\dagger \eta_1 \neq 0.
\ee
If $(e^{\ri \tau},g)$ is from the isotropy subgroup of this point, then we obtain that $e^{2 \ri \tau}=1$ must hold.
Next, if $e^{\ri \tau}$ was equal to $-1$, then $g=(g_1, g_2)$ must verify
$g_1 x g_2^{-1} = - x$ with $x$ in \eqref{S0}, which implies that $g$ is of the form
\be
g = \diag(z_1, \dots, z_n,  - z_1,\dots, - z_n)
 \ee
 with some $z_1,\dots, z_n \in U(1)$.
 But $g$ must also satisfy
 \be
 g_1 \zeta = \zeta \quad \hbox{and}\quad g_2 \eta = \eta.
 \ee
 By combining this with the assumption that the $\ell_1 \times \ell_2$ matrix $\zeta_1^\dagger \eta_1$ is not identically zero, we obtain a contradiction.
Therefore $e^{\ri \tau}=1$ must hold, and consequently $g=\1_{2n}$, since
the action of $K$ is free on $J^{-1}(\mu_0)$.
\end{proof}

\begin{rem}
If $l_1>0$, $l_2 =0$ and \eqref{regcond} holds, then the principal isotropy group with respect to the $U(1) \times K$ action
on $M^c$ contains two elements, the identity and $(-1, g_0)$ with $g_0 = \diag(\1_n, - \1_n)\in K$.
This means that the Hamiltonian evolution on $M^c$ under time $T/2$ coincides with a $K$ gauge transformation.
(This is  the case on the full phase space $M$, too.)
The reduced system is still strongly isochronous, but with period $T/2$.
\end{rem}

The Hamiltonian \eqref{HB2} seems to be closely related to the quantum mechanical model due to Yamamoto \cite{Y},
but as of writing the precise connection  is not clear.

\section{Discussion and outlook}
\label{S:7}

In this paper, we first exhibited the maximal superintegrability of strongly isochronous Hamiltonians from a group theoretic
viewpoint. Then, we considered reductions under compact symmetry groups and provided simple sufficient conditions
for guaranteeing that the reduced phase space is a smooth manifold on which the reduced
Hamiltonian has the same basic period as the original one.
As applications, we  proved the maximal superintegrability of the Lie algebraic spin Calogero--Moser
models \eqref{440}   as well as  of the Gibbons--Hermsen  model  \eqref{51}
and its $B_n$ type generalization \eqref{HB2}.
It should be noted that in each case the just  cited form of the reduced oscillator Hamiltonian is valid on a dense open subset
of the reduced phase space,
where the position variables satisfy a regularity condition, and
some periodic motions may leave this subset.

Building on the earlier works \cite{AKLM,Hoch}, one can construct a larger family of Lie algebraic spin Calogero--Moser models
with harmonic term as follows. Take any real simple Lie algebra $\fG0$ with associated Cartan decomposition
$\fG0 = \cK_0 + \cP_0$ and reduce the isotropic harmonic oscillator defined on the phase space $T^* \cP_0$ via a symmetry group
corresponding to the maximal compact subalgebra $\cK_0 < \fG0$.
From this perspective, the models of section \ref{S:4} are attached to the
complex simple Lie algebras  regarded as real Lie algebras.  Other distinguished special cases are
furnished by the split real simple Lie algebras. The simplest example is $\fG0 = \mathfrak{sl}(n,\bR)$, which yields
the Euler--Calogero--Moser model introduced in \cite{Wo2}.
Further generalizations based on  polar representations of  compact simple Lie groups
can also be derived \cite{AKLM}.

In section \ref{S:4}, we restricted ourselves to the principal orbit type for the compact symmetry group in order
to avoid technical complications. However, maximal superintegrability of the reduced
Hamiltonian is expected to hold on all symplectic leaves of the full Poisson quotient space,
except  trivial leaves like the one obtained from the equilibrium point of the unreduced oscillator.
The reason is that projections of periodic trajectories remain periodic, generically with the original period.
It is a largely open problem to describe the full stratified reduced phase spaces \cite{OR,SL} of the here-considered  and
other spin Calogero--Moser type models
that arise in the Hamiltonian reduction approach \cite{WGMPnew,Re2}.
A step forward in this direction has been taken
 in the recent article \cite{CJRX}.

 In the future, it would be interesting  to explore the detailed structure of the Poisson algebras formed by the constants of motion of the
 confining spin Calogero--Moser models.
 Due to a general theorem \cite{Sch}, one may restrict to $K$-invariant constants of motion of the original system
 that are polynomial in the unreduced variables.
 Quantum mechanical counterparts  of these Poisson algebras could be worth studying, too.
 Finally, one should also uncover the relationship between
 the quantization of the $B_n$ type confining spin Calogero--Moser Hamiltonian \eqref{HB2} and
the model introduced in \cite{Y}  directly at the quantum mechanical level.

\bigskip
\medskip
\subsection*{Acknowledgement}
I wish to thank Maxime Fairon for very useful discussions.
This work was supported in part by the NKFIH research grant K134946.

\end{document}